\documentclass{article}
\usepackage{amsmath,amsthm,amssymb,enumitem,relsize}
\usepackage[margin=2.5cm]{geometry}

\newtheorem{theorem}{Theorem}[section]
\newtheorem*{theorem*}{Theorem}
\newtheorem{definition}{Definition}[section]
\newtheorem*{definition*}{Definition}
\newtheorem{proposition}{Proposition}[section]
\newtheorem*{proposition*}{Proposition}

\newtheorem*{corollary*}{Corollary}
\newtheorem{lemma}{Lemma}[section]
\newtheorem*{claim*}{Claim}

\newcommand{\nid}{\mathrm{N}}

\newcommand{\id}{\mathrm{E}}
\newcommand{\lra}{\hspace{1pt}{\leftrightarrow}\hspace{1pt}}
\newcommand{\abs}[1]{\left|#1\right|}
\newcommand{\ftriv}{f_\textnormal{tr}}
\DeclareMathOperator{\K}{K}
\DeclareMathOperator{\C}{C}
\DeclareMathOperator{\floor}{floor}
\DeclareMathOperator{\ceil}{ceil}
\DeclareMathOperator{\osc}{osc}
\DeclareMathOperator*{\avg}{\mathlarger{\mathbb{E}}}
\renewcommand{\mid}{\,|\,}
\renewcommand{\O}{{O}}

\makeatletter
\newcommand{\eqnumk}{\tagform@{\texttt{k}} }
\makeatother

\newlength{\miniparindent}
\renewenvironment{samepage}{
  \noindent
  \begin{minipage}{\textwidth}
  \setlength{\parindent}{\miniparindent}
  \noindent
}{
  \end{minipage}
}

\usepackage{tikz}
\usetikzlibrary{shapes.arrows,chains,positioning}

\newcommand{\board}[3][1]{
   \foreach \x in {0,...,#2}
     \draw (\x,0) -- (\x,#3/2);
   \foreach \y in {0,...,#3}
     \draw (0,\y/2) -- (#2,\y/2);
     \draw (0, #3/2-0.25) node[anchor=east] {\small{$n-1$}}; 
     \draw (0, 0.25) node[anchor=east] {\small{$0$}}; 
     \node (N) at (#2/2,-0.6) {\small{$2^n$}};
     \draw[<-] (0, -0.6) -- (N);
     \draw[->] (N) -- (#2,-0.6);
     }

\newcommand{\wpawn}[2]{\draw ( #1 + 0.25 , #2/2 + 0.25 ) circle(1.5mm);}
\newcommand{\bpawn}[2]{\filldraw ( #1 + 0.75 , #2/2 + 0.25 ) circle(1.5mm);}

\begin{document}

\title{The normalized algorithmic information distance can not be approximated}
\author{Bruno Bauwens, Ilya Blinnikov\footnote{
   National Research University Higher School of Economics, 
   11, Pokrovsky Boulevard, 109028, Moscow. 
   }
   }
   \date{}
\maketitle

\begin{abstract}
  It is known that the normalized algorithmic information distance is not computable and not semicomputable.
  We show that for all $\varepsilon < 1/2$, there exist no semicomputable functions that differ from $\nid$ 
  by at most~$\varepsilon$. Moreover, for any computable function $f$ such that $| \lim_t f(x,y,t) - \nid(x,y)| \le \varepsilon$ 
  and for all $n$, there exist strings $x,y$ of length $n$ such that
  $\sum_t \abs{f(x,y,t+1) - f(x,y,t)} \ge \Omega(\log n)$. This is optimal up to constant factors.

  We also show that the maximal number of oscillations of a limit approximation of $\nid$ is~$\Omega(n/\log n)$. 
  This strengthens the $\omega(1)$ lower bound from [K. Ambos-Spies, W. Merkle, and S.A. Terwijn, 2019, 
  {\em Normalized information distance and the oscillation hierarchy}]. 
\end{abstract}

\section{Introduction}

The information distance defines a metric on bit strings that in some sense takes all ``algorithmic regularities'' into account. 
This distance was defined in~\cite{infoDistance} as $\id(x,y) = \max \{\K(x \mid y), \K(y \mid x)\}$, where $\K(\cdot| \cdot)$ denotes conditional prefix Kolmogorov complexity relative to 
some fixed optimal prefix-free Turing machine; 
we refer to appendix~\ref{sec:introKolm} for the definition and basic properties, 
and to the books~\cite{LiVitanyiForthEdition,bookShenVereshchagin} for more background.
After minor modifications, this distance satisfies the axioms of a metric, as explained in subsection~\ref{ss:infoDist}.
We refer to~\cite{idRevisited} for an overview of many equivalent characterizations.

The distance is not computable. 
However, conditional Kolmogorov complexity is {\em upper semicomputable}, which means that
there exists a computable function $f\colon \{0,1\}^* \times \{0,1\}^* \times \mathbb N \rightarrow \mathbb Q$ 
for which $\K(x \mid y) = \lim_t f(x,y,t)$, and that is non-increasing in its last argument~$t$.
Hence, also $\id$ is upper semicomputable.

The distance $\id$ is useful to compare strings of similar complexity.
However, for strings of different complexity, a normalized variant is often preferable. 

\begin{definition}\label{def:nid} The  {\em normalized} algorithmic information distance of strings $x$ and $y$ is\footnote{
  The numerator is nonzero, even if $x=y$.
  $\K_U(x) \ge 1$ holds for every choice of the optimal prefix-free Turing machine $U$,
  because such machines never halt on input the empty string.
  Indeed, if it halted, then it would be the only halting program by the prefix property, and hence, the machine can not be optimal.
  }
\[
  \nid(x,y) \;=\; \frac{\max \{\K(x \mid y), \K(y \mid x)\}}{\max \{\K(x), \K(y)\}}.
\]
\end{definition}

\noindent
This normalized distance has inspired many applications in machine learning, where complexities are heuristically estimated using 
popular practical compression algorithms such as gzip, bzip2 and PPMZ, see \cite[section~8.4]{LiVitanyiForthEdition}.
Within small additive terms, the function $\nid$ has values in the real interval $[0,1]$ and satisfies the axioms of a metric: 
\begin{itemize}[leftmargin=*]
  \item 
    $0 \le \nid(x,y) \le 1 + \O(1/\K(x,y))$, 
  \item 
    $\nid(x,y) = \nid(y,x)$, 
  \item 
    $\nid(x,x) \le \O(1/\K(x))$, 
  \item 
    $\nid(x,y) + \nid(y,z) \ge \nid(x,z) - \O( (\log \K(y))/\K(y))$.
\end{itemize}
See~\cite[Theorem 8.4.1]{LiVitanyiForthEdition}.\footnote{
  In~\cite[Exercise 8.4.3]{LiVitanyiForthEdition} it is claimed that for the prefix variant of the normalized information distance, 
  one can improve the precision of the last item to~$\O(1/\K(x,y,z))$. 
  However, we do not know a proof of this.  If this were true, then with minor modifications of $\nid$ similar to those in appendix \ref{ss:infoDist}, 
  all axioms of a metric can be satisfied precisely. 
}

\medskip
In this paper, we study the computability of~$\nid$. Note that if Kolmogorov complexity were computable, then also~$\nid$ would be computable. 
But this is not the case, and  
in~\cite{NIDnonapprox} it is proven that $\nid$ is not upper semicomputable and not lower semicomputable, (i.e. $-\nid$ is not upper semicomputable). 
Below in Lemmas~\ref{lem:nid_not_lowersemicomputable} and~\ref{lem:nid_not_uppersemicomputable} we present simple proofs. 
In fact, in~\cite{NIDnonapprox} it is proven that 
(i) there exists no lower semicomputable function that differs from $\nid$ by at most some constant $\varepsilon < 1/2$, and 
(ii) there exists no upper semicomputable function that differs at most~$\varepsilon = (\log n)/n$ from $\nid$ on $n$-bit strings.
Theorem~\ref{th:inapprox} below implies that (ii) is also true for all~$\varepsilon < 1/2$.

By definition, $\nid$ is the ratio of two upper semicomputable functions, and hence it is {\em limit computable}, which means that 
there exists a computable function $f$ such that $\nid(x,y) = \lim_t f(x,y,t)$.
A function $f$ that satisfies this property is called a {\em limit approximation} of~$\nid$.

We define a {\em trivial limit approximation} $\ftriv$ of $\nid$ where $\ftriv(x,y,t)$ 
is obtained by replacing all appearances of $\K( \cdot)$ and $\K( \cdot | \cdot)$ in  Definition~\ref{def:nid} 
by upper approximations $\K_t(\cdot)$ and $\K_t(\cdot | \cdot)$, where $(x,t) \mapsto \K_t(x)$ 
is a computable function satisfying $\lim \K_t(x) = \K(x)$ and $\K_1(x) \ge \K_2(x) \ge \ldots$; 
and similar for $\K_t(\cdot |\cdot)$. We assume that $\K_1(x \mid y)$ and $\K_1(x)$ are bounded by $\O(n)$ 
for all $x$ of length~$n$.

\begin{lemma}\label{lem:trivialApprox}
  For all $n$ and strings $x,y$ of length at most $n$:
  \[
  \sum_{t=1}^\infty \abs{\ftriv(x,y,t+1) - \ftriv(x,y,t)} \;\;\le\;\; 2\ln n + \O(1).
  \]
\end{lemma}

\begin{definition*}
  An {\em $\varepsilon$-approximation} of a function $g$ is a limit approximation of a function $g'$ 
  with $g - \varepsilon \le  g' \le g + \varepsilon$.
\end{definition*}

\noindent
For a suitable choice of $U$, we have $0 \le \nid \le 1$, and
the function defined by $f(x,y,t) = 1/2$ is a $(1/2)$-approximation.\footnote{
  For general optimal $U$, and for $\varepsilon > 1/2$, we can obtain an $\varepsilon$-approximation that is constant in $t$
  by choosing $f(x,y,t) = \nid(x,y)$ for some finite set of pairs $(x,y)$, and by choosing $f(x,y,t) = 1/2$ otherwise.
  }
We show that for $\varepsilon < 1/2$ and every $\varepsilon$-approximation, the sum in the above lemma is at least logarithmic.

\newcommand{\thInapprox}{
  Let $f$ be an $\varepsilon$-approximation of $\nid$ with $\varepsilon < 1/2$. For large $n$:
  \[
    \max_{x,y \in \{0,1\}^n} \;\sum_{t=1}^\infty \abs{f(x,y,t+1) - f(x,y,t)} \;\;\ge\;\; \tfrac{1}{100} \cdot (1-2\varepsilon)^2 \cdot \log n.
  \]
}
\begin{theorem}\label{th:inapprox}
  \thInapprox
\end{theorem}

\noindent
This result implies that for each $\varepsilon < 1/2$, there exists no upper semicomputable function that differs from $\nid$ by at most~$\varepsilon$.

\bigskip
\noindent
We now state the main result of~\cite{NIDoscillations}. 

\begin{definition*}
Let $k \ge 1$. A sequence $a_1, a_2, \ldots$ of real numbers has {\em at most $k$ oscillations} if the sequence can be 
  written as a concatenation of $k$ sequences ($k-1$ finite and 1 infinite) such that each sequence is either 
  monotonically non-increasing or non-decreasing. The sequence has $0$ oscillations if $a_1 = a_2 = \ldots$
\end{definition*}

\noindent
The main result of~\cite{NIDoscillations} states that no $0$-approximation $f$ of $\nid$ 
has at most a constant number of oscillations. More precisely, for each $k$, there exists a pair $(x,y)$ such that 
$f(x,y,1), f(x,y,2), \ldots$ does not have at most $k$ oscillations.

Let $k \colon \mathbb N \rightarrow \mathbb N$. We say that $f$ {\em has at least $k(n)$ oscillations}, 
if for all $n$ there exists a pair $(x,y)$ of strings of length at most $n$, 
such that $f(x,y,1), f(x,y,2), \ldots$ does not have at most $k(n)-1$ oscillations.
(The proof of) Theorem~\ref{th:inapprox}  implies that if $\varepsilon < 1/2$, 
then any $\varepsilon$-approximation has at least $\Omega((1-2\varepsilon)^2 \log n)$ oscillations. 

The trivial $0$-approximation $\ftriv$ has at most $\O(n)$~oscillations, because each upper-approximation of Kolmogorov complexity 
in its definition is bounded by $\O(n)$ on $n$-bit strings, and hence, there can be at most this many updates.
Can it be significantly less than~$n$, for example at most $n/100$ for large $n$?

The answer is positive. For all constants $c$, there exist optimal machines $U$ in the definition of complexity $\K$ for 
which the number of updates of $\K_t$ is at most $n/c + O(\log n)$.
For example, one may select an optimal $U$ whose halting programs all have length $0$ modulo~$c$.
If $\nid$ is defined relative to such a machine, than the total number of updates is $2n/c + O(\log n)$.
Hence, for every constant~$e$ there exists a version of~$\nid$ and a $0$-approximation 
that has at most~$n/e$ oscillations for large input sizes~$n$.
Our second main result provides an almost linear lower bound on the number of oscillations.

\newcommand{\propOscillations}{
  Every $0$-approximation of $\nid$ has at least $\Omega(n/\log n)$ oscillations.
}

\begin{theorem} \label{th:oscillations}
  \propOscillations
\end{theorem}


\noindent
In an extended version of this article, we plan to improve the $\Omega(n/\log n)$ lower bound to an $\Omega(n)$ bound. 
This requires a more involved variant of our proof.

Theorems~\ref{th:inapprox}  and~\ref{th:oscillations} both imply that $\nid$ and hence Kolmogorov complexity is not computable.
In fact, they imply something stronger:
$\K(\K(x \mid y) \mid x,y)$ can not be bounded by a constant.\footnote{ 
  Indeed, if this were bounded by~$c$, 
  there would exist an upper approximation $f$ of $\K(\cdot \mid \cdot)$
  such that for each pair $(x,y)$, the function $f(x,y,\cdot)$ has only finitely many values.
  (We modify any upper approximation of complexity by only outputting
  values $k$ on input $x,y$, for which $\K(k \mid x,y) \le c$. There are at most $2^c$ such~$k$.)
  Hence, there would exist an approximation $f'$ of $\nid$ such that for all $x, y$, the function $f'(x,y, \cdot)$ has only finitely many values.
  Such functions would have only finitely many oscillations, contradicting  Theorem~\ref{th:oscillations},
  and a finite total update, contradicting  Theorem~\ref{th:inapprox}.
  }
It has been shown that $\K(\K(x) \mid x) \ge \log n - O(1)$, see~\cite{complexityOfComplexity,compcomp},
and our proofs are related.
Like the proof in~\cite{compcomp}, we also use game technique.
This means that we present a game, present a winning strategy, and show that this implies the result.
The game technique often leads to tight results with more intuitive proofs.
(Moreover, the technique allows to easily involve students in research, because after the game is formulated, 
typically no specific background is needed to find a winning strategy.)
For more examples of game technique in computability theory and algorithmic information theory, 
we refer to~\cite{KolmogorovGames}.

\section*{$\nid$ is not upper nor lower semicomputable}

For the sake of completeness, we present short proofs of the results in~\cite{NIDnonapprox}, 
obtained from Theorem 3.4 and Proposition 3.6 from~\cite{NIDoscillations} (presented in a
form that is easily accessible to people with little background in the field).\footnote{
  NOTE TO THE REVIEWER: papers about the information distance are sometimes cited by people from more applied research areas.
  In an optimistic scenario, there might exist such readers that want to read some initial segment of the paper, 
  and might not remember the proof of the uncomputability of $\K(\cdot)$. 
  Hence, I think it is good to keep the proof of Lemma~\ref{lem:K_has_no_lowerbound}. 
}
A function $g$ is  {\em lower semicomputable} if $-g$ is upper semicomputable. 

\begin{lemma}\label{lem:nid_not_lowersemicomputable}
$\nid$ is not lower semicomputable. 
\end{lemma}

\begin{proof}
  Note that for large $n$, there exist $n$-bit $x$ and $y$ such that 
  \[
  \nid(x, y) \;\ge\; 1/2\,.
  \]
  Indeed, for any $y$, there exists an $n$-bit $x$ such that $\K(x \mid y) \ge n$. The denominator of $\nid$ is at most $n + O(\log n)$, and the inequality follows for large~$n$.

  Assume $\nid$ was lower semicomputable. On input $n$, 
  one could search for such a pair $(x,y)$, and we denote the first such pair that appears by $(x_n, y_n)$. 
  We have $\K(x_n) = \K(n)+\O(1)$
  and $\max\{\K(x_n \mid y_n), \K(y_n \mid x_n)\} \le \O(1)$. Hence $\nid(x_n,y_n) \le \O( 1/\K(n))$.
  For large $n$ this approaches~$0$, contradicting the equation above.
\end{proof}

\noindent
{\em Remark.} With the same argument, it follows that for any $\varepsilon < 1/2$, 
there exists no lower semicomputable function that differs from $\nid$ by at most~$\varepsilon$. 
Indeed, instead of $\nid(x,y) \ge 1/2$ we could as well use $\nid(x,y) \ge 1/2 + \varepsilon$, 
and search for $(x_n, y_n)$ for which the estimate is at least~$1/2$.

\medskip
\noindent
To prove that $\nid$ is not upper semicomputable, we use the following well-known lemma.

\begin{lemma}\label{lem:K_has_no_lowerbound}
The complexity function $\K(\cdot)$ has no unbounded lower semicomputable lower bound. 
\end{lemma}

\begin{proof}
  This is proven by the same argument as for the uncomputability of $\K$, see appendix \ref{ss:introC}:
  suppose such bound $B(x) \le \K(x)$ exists. Then on input~$n$, one can search for a string~$x_n$ 
  with $n \le B(x_n)$ and hence~$n \le \K(x_n)$. But since there exists an algorithm to compute~$x_n$ given~$n$, 
  we have $\K(x_n) \le  \O(\log n)$. 
  This is a contradiction for large~$n$. Hence, no such $B$~exists.
\end{proof}

\begin{lemma}\label{lem:nid_not_uppersemicomputable}
  $\nid$ is not upper semicomputable. 
\end{lemma}

\begin{proof}
  By optimality of the prefix-free machine in the definition of $\K$, 
  we have that $\K(x \mid y) \ge 1$ for all~$x$ and~$y$.
  Thus $1 \le \K(x \mid x) \le O(1)$, and hence,
  \[
    1/\K(x) \;\le\; \nid(x,x) \;\le\; \O(1/\K(x)).
  \]
  If $\nid$ were upper semicomputable, we would obtain an unbounded lower semicomputable 
  lower bound of~$\K$, which contradicts Lemma~\ref{lem:K_has_no_lowerbound}.
\end{proof}

\section{Trivial approximations have at most logarithmic total update}

Lemma~\ref{lem:trivialApprox} follows from the following lemma for $c \le \O(1)$ 
and the upper bound $m \le O(n)$ on the upper approximations of Kolmogorov complexity.

\begin{lemma}\label{lem:upperbound}
  Assume $1 \le a_1 \le a_2 \le \dots \le a_m \le m$, $1 \le b_1 \le b_2 \le \dots \le b_m \le m$ and $a_i \le b_i + c$.
  Then, 
  \[
    \sum_{i \le m} \left| \frac{a_i}{b_i} - \frac{a_{i+1}}{b_{i+1}} \right| \;\;\le\;\; 2\ln m + \O(c^2).
  \]
\end{lemma}

\begin{proof}
 We first assume $c=0$. We prove a continuous variant. 
 Let $\alpha,\beta \colon [0,m] \rightarrow [1,m]$ be non-decreasing real functions with $\alpha(t) \le \beta(t)$ and $1 \le \alpha(0) \le \beta(m) \le m$. 
 The sum in the lemma can be seen as a special case of
 \[
 \int_{t=0}^{t=m} \left|\textnormal{d}\frac{\alpha(t)}{\beta(t)}\right| \;\; = \;\; \int \frac{\textnormal{d}\alpha(t)}{\beta(t)} \;\;+\;\; \int \frac{\alpha(t)}{\beta^2(t)}\textnormal{d}\beta(t).
 \]
 The left integral in the sum is maximized by setting $\beta(t)$ equal to its minimal possible value, which is $\alpha(t)$.
 The right one is maximized for the maximal value of $\alpha(t)$, which is~$\beta(t)$.
 Thus,
 \[
 \le \;\;\int_{u=\alpha(0)}^{u=\alpha(m)} \frac{\textnormal{d}u}{u} \;\; + \;\;
 \int_{u=\beta(0)}^{u=\beta(m)} \frac{\textnormal{d}u}{u} \;\;\le\;\; 2\ln m.
 \]
  For $c \ge 0$, the minimal value of $\beta$ is $\max \{1,\alpha - c\}$ and the maximal value of $\alpha$ is $\min \{m, \beta + c\}$.
  The result follows after a calculation.
\end{proof}

%

\section{Oscillations of $0$-approximations, the game}

For technical reasons, we first consider the {\em plain length conditional} variant of the normalized information distance $\nid'$.
For notational convenience, we restrict the definition to pairs of strings of equal length.

\begin{definition*}
  For all $n$ and strings $x$ and $y$ of length $n$, let
  \[
    \nid'(x,y) \;=\; \frac{\max \left\{ \C(x\mid y), \C(y\mid x) \right\}}{\max \{\C(x\mid n), \C(y\mid n)\}}.
  \]
  If $\C(x \mid n) = 0$, let $\nid'(x,x) = 0$.
\end{definition*}

\noindent
{\em Remarks.}
\\- For $x \ne y$, the denominator is at least $1$, since at most $1$ string can have complexity zero relative to~$n$.
\\- The choice of the value of $\nid'(x,x)$ if $\C(x \mid n) = 0$ is arbitrary, and does not affect Proposition~\ref{prop:oscillationPlain} below.
\\- In the numerator, the length $n$ is already included in the condition, since it equals the length of the strings.
\\- There exists a trivial approximation of $\nid'$ with at most $2n + O(1)$ oscillations. 
Indeed, consider an approximation obtained by defining $\C_t(\cdot | \cdot)$ with brute force searches among programs of length at most~$n + O(1)$. 
\\- Again, for every constant~$e$, we can construct an optimal machine and a $0$-approximation of~$N'$ for which the number of 
oscillations is at most~$n/e$. We now present a matching lower bound.

\begin{proposition}\label{prop:oscillationPlain}
  Every $0$-approximation of $\nid'$ has at least $\Omega(n)$ oscillations.
\end{proposition}

\noindent
In this section, we show that the proposition is equivalent to the existence of a
winning strategy for a player in a combinatorial (full information) game.  
In the last section of the paper, we present such a winning strategy.  

\medskip
\noindent
{\em Description of game $\mathcal G_{n,c,k}$.} 
The game has $3$ integer parameters: $n \ge 1$, $c \ge 1$ and $k \ge 0$. 
It is played on two 2-dimensional grids $\mathrm X$ and $\mathrm Z$. 
Grid $\mathrm X$ has size $n \times 2^n$. Its rows are indexed by integers $\{0,1,\ldots,n-1\}$,
and its columns are indexed by $n$-bit strings. Let $\mathrm X_u$ be the column indexed by the string $u$.
See figure~\ref{fig:gridX} for an example with $n=3$. 
Grid $\mathrm Z$ has size $n \times {2^n + 1 \choose 2}$. 
The rows are indexed by integers $\{0,\ldots,n-1\}$,
and its columns are indexed by unordered pairs $\{u,v\}$, where $u$ and $v$ are $n$-bit strings, 
(that may be equal).\footnote{
  Formally, we associate sets $\{u,v\}$ with 2 elements to an unordered pair $(u,v)$, and 
  singleton sets $\{u\}$ to the pair $(u,u)$.
}
We sometimes denote unordered pairs $\{u,v\}$ of $n$-bit strings as $uv$, and write $\mathrm Z_{\{u,v\}} = \mathrm Z_{uv}$. 
Note that $\mathrm Z_{uv} = \mathrm Z_{vu}$.
Let $u \in \{0,1\}^n$. The {\em slice} $\mathrm Z_{u}$ of~$\mathrm Z$ is the 2-dimensional grid of size $n \times 2^n$ containing all 
 columns $\mathrm Z_{uv}$ with $v \in \{0,1\}^n$.
Additionally, Bob must generate a function $f$ mapping unordered pairs of $n$-bit strings and natural numbers to real numbers.

\begin{figure}[h]
  \centering
  \begin{tikzpicture}
    \board{8}{3}

    \node[anchor=north] at (0.5, 0) {\footnotesize{000}};
    \node[anchor=north] at (1.5, 0) {\footnotesize{001}};
    \node[anchor=north] at (2.5, 0) {\footnotesize{010}};
    \node[anchor=north] at (3.5, 0) {\footnotesize{011}};
    \node[anchor=north] at (4.5, 0) {\footnotesize{\dots}};

    \wpawn{0}{2}
    \wpawn{1}{2}
    \bpawn{1}{0}
    \bpawn{0}{2}
    \wpawn{3}{2}
  \end{tikzpicture}
  \caption{\small{Example of board $\mathrm X$ with $n=3$. Alice has placed 2 tokens in row 2 (white), and Bob has placed 1 token in row 0 and 1 in row 2 (black).
  The row restrictions for both players are satisfied, since $\max\{1,3\} \le 2^2$ and $1 \le 2^0$. $X_{000} = X_{011} = 2$, $X_{001} = 0$ and $X_{010}=3$.}}
  \label{fig:gridX}
\end{figure}
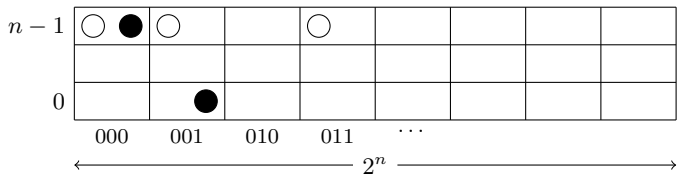

Two players, Alice and Bob, alternate turns.  
The rounds are numbered as $t = 1, 2, \ldots$
At each round, Alice plays first. At her turn, she places tokens on cells of the grids. 
She must place at least 1 token. 
Afterwards, Bob places zero or more tokens on the grids, and 
he declares all values $f(uv,t)$ for all unordered pairs $\{u,v\}$,
where $t$ is the number of the current round. 
This terminates round $t$, and the players start with round $t+1$.

For each player, for each $i \in \{1,\ldots, n\}$, and for all grids $\mathrm G \in \{\mathrm X \} \cup \{\mathrm Z_{u} : u \in \{0,1\}^n\}$, 
the following  {\em row restriction}  should be satisfied:
{\em The total number of tokens that the player has placed during the whole game in the $i$-th row of $\mathrm G$, is at most~$2^i$.}
If a player does not satisfy this restriction, the game terminates and the other player {\em wins}. See figure~\ref{fig:gridX}.
Bob's moves should satisfy 2 additional requirements. If after his turn 
these requirements are not satisfied, the game terminates and Alice wins.
\begin{itemize}[leftmargin=*]
  \item 
    Let $X_u$ be the value of column $\mathrm X_u$ given by the minimal row-index of a cell in $\mathrm X_u$ containing a token.
    If $\mathrm X_u$ contains no tokens, then $X_u = n$.
    Similar for the value~$Z_{uv}$ of column~$\mathrm Z_{uv}$.
    For all $u$ and $v$:
    \begin{equation}\tag{\texttt{c}}\label{eq:requirement_c}
      \frac{Z_{uv}-1}{\max \{X_u, X_v\}+c} \;<\; f(uv,t) \;<\; \frac{Z_{uv}+c}{\max \{X_u, X_v\}}.
   \end{equation}
    
  \item 
    For all $u$ and $v$: 
    $f(uv,1), f(uv,2), \ldots$ has at most $k$ oscillations.
    \hfill \eqnumk\!\!
\end{itemize}
\noindent
Note that for decreasing $c$ and $k$, it becomes easier for Alice to win. 

\medskip
\noindent
{\em Discussion.} 
If Alice places a token in a row with 
small index, Bob has a dilemma: either he can change the function~$f$, 
or he can place tokens on the other board to restore the ratios in~\eqref{eq:requirement_c}.
In the first case, he might increase the number of oscillations in~\eqnumk\!\!, while 
in the second case, he exhausts his limited capacity to place 
tokens on rows of small indices, (by the row restriction, at most $1 + 2^1 + \ldots + 2^{i-1} = 2^i - 1$ 
tokens can be placed below row~$i$ in each grid~$G$).

\medskip
\noindent
{\em Remark.} 
The game has at most $\O(n 2^{2n})$ rounds, because 
in each round, Alice must place at least 1 token, and by the row restriction, Alice 
can place at most $\O(n 2^{2n})$ tokens on all grids.
Hence, the game above is finite and has full information.
This implies that either Alice or Bob has a winning strategy. 


\begin{lemma}\label{lem:zeroApproximationGivesWinning}
  Let $k \colon \mathbb N \times \mathbb N \rightarrow \mathbb Z$ be such 
  that Alice has a winning strategy in the game~$\mathcal G_{n,c,k(n,c)}$.
  Then for every 0-approximation of $\nid'$ there exists a constant $c$ such that 
  for large $n$, the 0-approximation has more than $k(n,c)$ oscillations on $n$-bit inputs.
\end{lemma}

\begin{proof}
  The idea of the proof is to use any limit approximation $f'$ to construct a strategy for Bob. 
  By assumption there exists some winning strategy for Alice, and we let it play against this strategy for Bob.
  Then we show that Bob satisfies the row restrictions and requirement  \eqref{eq:requirement_c}. 
  Since Alice strategy is winning, we conclude that requirement \eqnumk must be violated. 
  Our construction implies that $f$ has fewer oscillations then~$f'$, thus also $f'$ has more than~$k(n,c)$ oscillations.

  It suffices to prove the lemma for the largest function $k(n,c)$ for which Alice wins the game $\mathcal G_{n,c,k(n,c)}$. 
  This function $k$ is computable, since the game is finite, and for each value we can determine whether Alice has a winning 
  strategy by brute force searching all strategies.
  \\- Let $\C_s(\cdot| \cdot)$ represent an upper approximation of~$\C(\cdot \mid \cdot)$. 
  \\- Let $\C(u \lra v) = \max \{\C(u \mid v), \C(v \mid u) \}$ and similar for~$\C_s( u \lra v)$.
  \\- Let $f'$ be a $0$-approximation of~$\nid'$. Without loss of generality, we assume $f'(u,v,t) = f'(v,u,t)$.
  
  \bigskip
  \noindent
  For all $c$ and~$n$, we present a run of the game $\mathcal G_{n,c,k(n,c)}$. The mapping from $c$ and~$n$ 
  to a (transcript of) this run is computable.
  First, we fix a winning strategy of Alice in the game $\mathcal G_{n,c,k(n,c)}$ in a computable way.
  For example, we may brute force search all strategies and select the first winning strategy that appears.
  Let~$r_0 = 1$. Consider the game in which Alice plays this strategy, and Bob replies as follows.

  \medskip
  \noindent
  \begin{minipage}{\textwidth}
  {\em Bob's strategy.} 
  At round $t$, Bob searches for a value $s$ with $s > r_{t-1}$ such that for all $u$ and $v$:
  \begin{itemize}
    \item[(i)]   $\C_s(u \mid n) \,<\, X_u + c$\;\, and \; $C_s(u \lra v) \,<\, Z_{uv} + c$, 
    \item[(ii)]  
    $  f'(u,v,s) \,=\, \frac{\C_s(u \lra v)}{\max \{ \C_s(u \mid n), \C_s(v \mid n) \}}$.
  \end{itemize}
  If such an $s$ is found, he sets $r_t = s$ and $f(uv,t) = f'(u,v,s)$ for all $u$ and~$v$. 
  For all $u$ he places a token in column $\mathrm X_u$ at row $\C_s(u \mid n)$. 
  For all unordered pairs $\{u,v\}$, he places a token in column~$\mathrm Z_{uv}$ at row~$\C_s(u \lra v) + 1$. 
    {\em End of Bob's strategy.} 
  \end{minipage}

  \bigskip
  \noindent
  We first show that if Bob does reply, he satisfies the row restriction.
  For $\mathrm G = \mathrm X$ this holds because there are at most $2^i$ programs of length~$i$, 
  and hence, at most $2^i$ strings $u$ with $\C_s(u) = i$ for some~$s$.
  For $\mathrm G = \mathrm Z_u$, this holds
  because  $\C_s(u \lra v) = i$ implies $\C(v \mid u) \le i$, and there are less than $2^{i+1}$ such~$v$.

  Assuming that Bob plays in round $t$, requirement~\eqref{eq:requirement_c} holds.
  Indeed, after Bob's move and for $s=r_t$, condition (i) implies:
    \[
    X_u \,\le\, \C_s(u \mid n) \,<\, X_u + c 
    \qquad \text{and} \qquad
    Z_{uv}-1 \,\le\, \C_s(u \lra v) \,<\, Z_{uv} + c. 
  \]
  Together with (ii) and $f(t,u,v) = f'(s,u,v)$, this implies requirement~\eqref{eq:requirement_c}.

  We show that for large $c$, there always exists an $s$ such that (i) and (ii) are satisfied, and hence, Bob plays in each round.
  Since $f'$ is a $0$-approximation, requirement (ii) is true for large $s$, and this does not depend on~$c$.
  We show that (i) is also satisfied. To prove the left inequality, we first construct a Turing machine~$M$. 
  The idea is that the machine plays the game above, and each time Alice places a token in a cell of column~$\mathrm{X}_u$ with row index~$i$, 
  it selects an unassigned $i$-bit string, and assigns to it the output~$u$.  Thus on input a string $p$ and integers $c,n$, 
  it plays the game, waits until the $p$-th token is placed in the row with index equal to the length of~$p$, and it outputs the column's index, 
  (which is an $n$-bit string).
  The row restriction implies that enough programs are available for all tokens.
  Hence, $\C_{M}(u \mid n,c) \le i$, whenever Alice places a token in $\mathrm X_u$ at height~$i$. 
  By optimality of the Turing machine in $\C(\cdot \mid \cdot)$, 
  we have $\C(u \mid n,c) \le X_u + O(1)$ for all $u$, and hence,  
  \[
    \C(u \mid n) \;\le\; X_u + O(\log c). 
  \]
  For large $c$, this is less than $X_u + c$. 
  By a similar reasoning, we have $\C(u \lra v) < Z_{uv} + c$, because each time Alice places a token in row $i$ of column $\mathrm Z_{uv}$, 
  we assign 2 programs of length $i$: one that outputs $u$ on input $v,n,c$, and one that outputs $v$ on input $u,n,c$.
  Thus, for large $s$, also requirement (i) is satisfied, and Bob indeed plays at any given round, 
  assuming he played in all previous rounds.

  \smallskip
  Recall that Alice plays a winning strategy, and that Bob satisfies the row restriction and requirement~\eqref{eq:requirement_c}.
  Hence, requirement~\eqnumk must be violated, i.e., for some pair $(u,v)$, the
  sequence $f(uv,1), f(uv,2), \ldots$ has more than $k(n)$ oscillations.
  Since $r_t$ is increasing in~$t$, this sequence is a subsequence of $f'(u,v,1), f'(u,v,2), \ldots$, and the latter must also have more than~$k(n)$ oscillations.
  This implies the lemma.
\end{proof}

\medskip
\noindent
To prove Theorem~\ref{th:oscillations} we need a version of the previous lemma for the prefix distance.

\begin{lemma}\label{lem:zeroApproximationGivesWinningVar}
  Under the assumption of Lemma~\ref{lem:zeroApproximationGivesWinning},
  every 0-approximation of $\nid$ 
  has more than $k(n,5\log n)$ oscillations on $n$-bit inputs for large~$n$.
\end{lemma}

\begin{proof}
  As a warm up, we observe that
  \[
    \K(x) \;\le\; \C(x \mid n) + 4 \log n + O(1). 
  \]
  Indeed, we can convert a program on a plain machine that has access to~$n$, to a program on some prefix-free machine without access to~$n$, 
  by prepending prefix-free codes of the integers $n$ and $\C(x \mid n)$. Each such code requires $2\log n + O(1)$ bits, 
  and hence the inequality follows.

  We modify the proof above by replacing all appearances of $\C(x \mid n)$ by $\K(x)$, of $\C(x \mid y)$ by $\K(x \mid y)$, 
  and similarly for the approximations $\C_s(\cdot \mid \cdot)$.
  We also set $c = 5\log n$ and assume that $f'$ is a $0$-approximation of~$\nid$.
  In Bob's strategy, no further changes are needed.  

  The row restriction for Bob is still satisfied, 
  because the maximal number of halting programs of length~$i$ on a prefix-free machine is still at most~$2^i$.
  Requirement  \eqref{eq:requirement_c} follows in the same way from items (i) and (ii) in Bob's strategy.
  It remains to prove that for large $c$ and $s$, these conditions (i) and (ii) are satisfied.
  Item (ii) follows directly, since $f'$ is a $0$-approximation of~$\nid$.

  For item (i), we need to construct a prefix-free machine $M'$. This is done in a similar way as above, 
  by associating tokens in row $i$ to programs of length $i$, but we also need to prepend 3 prefix-free codes: for the row index, for~$n$, and for~$c$.
  This implies 
  \[
    \K(u) \le X_u + 4 \log n +  O(\log c) .
  \]
  Recall that $c = 5 \log n$. Hence, this is at most $X_u + c$ for large~$n$. 
  The lemma follows from the violation of requirement~\eqnumk in the same way as before.
\end{proof}

\section{Total update of $\varepsilon$-approximations, the game}

We adapt the game for the proof of Theorem~\ref{th:inapprox}.

\medskip
\noindent
{\em Description of game $\mathcal H_{n,\varepsilon,a}$}, where $\varepsilon > 0$ and $a \ge 0$ are real numbers. 
The game is the same as the game of the previous section, except that requirements  \eqref{eq:requirement_c} and \eqnumk are replaced by: 
\begin{itemize}[leftmargin=*]
  \item 
    For all $u$ and $v$ with $\max \{X_u, X_v\} \ge \sqrt{n}$:
    \begin{equation}\tag{${\epsilon}$}\label{eq:requirement_eps}
      \left| f(u,v,t) - \frac{Z_{uv}}{\max \{X_u, X_v\}} \right| \;\;\le\;\; \varepsilon.
   \end{equation}
    
  \item 
    For all $u$ and $v$ with $\max \{X_u, X_v\} \ge \sqrt{n}$:
    \begin{equation}\tag{\texttt{a}}\label{eq:requirement_a}
      \sum_{s =1}^{t-1} |f(u,v,s) - f(u,v,s+1) | \;\;\le\;\; a.
    \end{equation}
\end{itemize}

\noindent
{\em Remarks.} 
\\- We call the sum in~\eqref{eq:requirement_a}, the {\em total update} of~$f$. Similar for the total update of an $\varepsilon$-approximation.
\\- The threshold $\sqrt{n}$ is chosen for convenience. Our proof also works with any computable threshold function that is at least super-logarithmic and at most $n^\alpha$ 
for some $\alpha < 1$.

\medskip

\begin{lemma}\label{lem:epsApproximationGivesWinning}
  Let $a\colon \mathbb N \rightarrow \mathbb R$. 
  Suppose that for large~$n$, Alice has a winning strategy in the game~$\mathcal H_{n,\varepsilon,a(n)}$.
  Fix $\varepsilon' < \varepsilon$, and an $\varepsilon'$-approximation $f'$ of either~$\nid'$ or~$\nid$.
  Then, for large~$n$, there exist $n$-bit inputs for which the total update of $f'$ exceeds~$a(n)$.
\end{lemma}

\begin{proof}
  We first consider an $\varepsilon'$-approximation $f'$ of $\nid'$, and at the end of the proof we explain the modifications for~$\nid$.
  The proof has the same high-level structure as the proof of Lemma~\ref{lem:zeroApproximationGivesWinning}: 
  from $f'$ we obtain a strategy for Bob that is played against Alice's winning strategy. Then, from the violation 
  of \eqref{eq:requirement_a} we conclude that the total update of $f'$ exceeds~$a(n)$.

  \medskip
  Let $n$ be large such that Alice has a winning strategy in the game $\mathcal H_{n,\varepsilon,a(n)}$.
  We consider a run of the game where Alice plays a computably generated winning strategy and Bob's replies are as follows.

  \medskip
  \begin{samepage}
   {\em Bob's strategy.} 
  He searches for an $s > r_{t-1}$ such that for all $u$ and $v$ with $\max\{\C_s(u), \C_s(v)\} \ge \sqrt{n}$:
  \begin{itemize}
    \item[(i)]   $\C_s(u \mid n) \le X_u + c$\;\, and \; $C_s(u \lra v) \le Z_{uv} + c$,  
    \item[(ii)]  
      $  \Big| f'(u,v,s) - \frac{\C_s(u \lra v)}{\max \{ \C_s(u \mid n), \C_s(v \mid n) \}} \Big| \;\le\; \varepsilon' $,
  \end{itemize}
  If such an $s$ is found, let $r_t = s$. Bob chooses $f(uv,t) = f'(u,v,s)$ for all $u$ and $v$.
  For all $u$ he places a token in column $\mathrm X_u$ at row $\C_s(u \mid n)$. 
  For all unordered pairs $\{u,v\}$, he places a token in column~$\mathrm Z_{uv}$ at row~$\C_s(u \lra v) + 1$. 
   {\em End of Bob's strategy.} 
  \end{samepage}

   \medskip
   \noindent
  For similar reasons as above, we have that for some $c$ and for large $s$, requirements (i) and (ii) are satisfied.
  This implies that for some $c$, Bob always reacts.

  We now verify that for large $n$, requirement \eqref{eq:requirement_eps} holds.
  Recall that we need to check the inequality when the denominator is at least $\sqrt{n}$.
  After Bob's move we have again that
  \begin{equation}\tag{*}\label{eq:XvsC}
    X_u \,\le\, \C_s(u \mid n) \,<\, X_u + c 
    \qquad \text{and} \qquad
    Z_{uv}-1 \,\le\, \C_s(u \lra v) \,<\, Z_{uv} + c. 
  \end{equation}
  Since $\nid' \le e$ for some constant $e$, we may also assume that $f' \le e$, because truncating $f'$ can only decrease the number of oscillations.
  This and item (ii) imply that if $n$ is large enough such that 
  \begin{equation}\tag{**}\label{eq:threshold_n}
    (c+1)\frac{e+1}{\sqrt{n}} \; \le \; \varepsilon - \varepsilon', 
  \end{equation}
  inequality  \eqref{eq:requirement_eps} is indeed satisfied.

  Because Bob loses, requirement \eqref{eq:requirement_a} must be violated.
  Since the total update of $f$ is at least the total update of $f'$ as long as the $\sqrt{n}$-threshold is not reached, this implies 
  that every $\varepsilon'$-approximation has total update more than~$a(n)$. The statement for~$\nid'$ is proven.

  \smallskip
  The modifications for $\nid$ are similar as in the previous section. 
  Instead of choosing $c$ to be a constant, we again choose it to be~$5\log n$, 
  and for the same reasons as above, this makes~\eqref{eq:XvsC} true 
  if we replace conditional plain complexity by (conditional) prefix complexity.
  This increase from constant to logarithmic $c$ increases the minimal value of $n$ in \eqref{eq:threshold_n} only 
  by a factor~$O(\log^2 n)$. Otherwise, nothing changes in the above argument. The lemma is proven.
\end{proof}

\section{Total update of $\varepsilon$-approximations, winning strategy}

\begin{lemma}\label{lem:stratH}
  Let $\varepsilon < 1/2$. For large $n$, Alice has a winning strategy in the game~$\mathcal H_{n, \varepsilon , \rho \log n}$ for
  \[
    \rho \;=\; \frac{1-2\varepsilon}{12 \log \tfrac{10}{1-2\varepsilon}}.
  \]
\end{lemma}

\noindent
By Lemma~\ref{lem:epsApproximationGivesWinning}, this implies  Theorem~\ref{th:inapprox}.

\begin{proof}[Proof idea.]
  Alice's winning strategy maintains a product set $U' \times V'$ containing pairs of strings. Initially, $U'$ and $V'$ are disjoint subsets of $\{0,1\}^n$ of size $2^{n-1}$.
We force Bob to decrease $f(u,v,\cdot)$ by at least a constant $\delta>0$ for a significant fraction of pairs $(u,v) \in U' \times V'$. 
Afterwards, we discard parts of $U'$ and of $V'$ such that for all pairs of the remaining set $U' \times V'$, 
this increase and decrease indeed happened. After a 'reset'-operation, we repeat the procedure.
We show that we can repeat this logarithmically many times before the sets $U'$ and $V'$ have size less than $2^{\sqrt{n}}$. And this implies the result.

The idea to enforce a decrease is as follows. First we consider a set $S = \bigcup_{j \le E} U_j \times V_j$, 
where $U_1, \ldots, U_E$ is any collection of pairwise disjoint subsets of $U'$ of some small size. 
The first time, we choose the size to be roughly $2^{\delta n}$
for some small constant $\delta$, and the number $E$ of sets $U_j$ equals roughly $2^{(1-\delta)n}$. 
  The sets $V_1, \ldots, V_E$ also 
  have size $2^{\delta n}$ and are chosen such that $Z_{uv}$ is larger than $n(1-\delta)$ for all $(u,v) \in U_j \times V_j$. 
This is possible, since the size of $U_j$ are very small, see figure~\ref{fig:strategy}. This implies that $f(uv,t)$ is very close to $1$ in~$S$. 

\begin{figure}\label{fig:strategy}
  \centering
  \begin{tikzpicture}[scale=0.5]
  \draw (0,0) rectangle (8,8);
  \foreach \i/\j in {0/7, 1/6.5, 2/5, 3/2.5, 4/3.3, 5/1.8}{
    \filldraw[black, fill=lightgray] (\i,\j) rectangle +(1,1);
  }
  \foreach \i in {0.3,1.6,3.1,5.4}{
    \draw[dashed,darkgray] (\i,0) -- (\i,8);
    \filldraw[darkgray] (\i,0) circle (2pt);
  }
    \node[anchor=north] at (4,0){$U'$};
    \node[anchor=east] at (0,4){$V'$};
\end{tikzpicture}
  \caption{\small{The set $S$ from Alice's strategy. For 4 indices~$j$, Bob reacted by decreasing $X_u$ for some $u \in U_j$. For 2 indices this did not happen, and they might be selected by Alice to start the next iteration.}}
\end{figure}
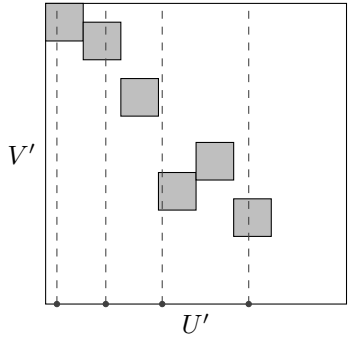

Alice now decreases $Z_{uv}$ for all $(u,v) \in S$ to $\delta n$. $\delta$ is chosen small enough such that if Bob wants to satisfy
requirement \eqref{eq:requirement_eps}, he is forced to 
  either decrease $X_{u}$ to less than $n(1-\delta)$ or to decrease $f$ by at least $\delta$. 
Since he can do the first only for a small fraction of strings~$u$, (for less than $2^{(1-\delta)n} \simeq E$ strings~$u$), 
there will be a part $U_j \times V_j$ that only contains $(u,v)$ for which the second option was chosen.
Afterwards, Alice decreases $X_u$ for all $u \in U_j \cup V_j$, and the procedure can be repeated as if the game was played for $n \leftarrow \delta n$.

In each iteration, the parameter decreases by a constant factor~$\delta$, and hence the strategy can be repeated logaritmically many times. 
  Hence, Alice can enforce a total update proportional to $\delta \cdot \log_{1/\delta} n$, and the proof overview.
\end{proof}

\noindent
We present the details.
The following technical lemma presents the set~$S$, from which a part $U_j \times V_j$ will be chosen for the recursion.

\begin{lemma}\label{lem:main_help}
  Let 
  \\- $E$ and $N$ be powers of $2$ with $E \le N/2$.
  \\- $U'$ and $V'$ be subsets in $\{0,1\}^n$ of size~$N$.
  \\- $U_1, \ldots, U_E$ be pairwise disjoint subsets of $U'$ of size~$N/(2E)$.
  \\There exist sets $V_1, \ldots, V_E$ such that
  for all $j \le E$ and $(u,v) \in U_j \times V_j$ we have $Z_{uv} \ge \log E$. 
\end{lemma}

\begin{proof}
  For each $u$, less than $E$ strings $v$ satisfy $Z_{uv} < \log E$.
  Fix some set~$U_j$. We need to select $V_j$. How many $v$ in $\{0,1\}^n$ satisfy 
  \[
   \text{ 
    $Z_{uv} < \log E$ \; for some \; $u \in U_j$\;? 
   }
  \]
  There are less than $(N/(2E)) \cdot E = N/2$ such $v$.
  Let $V_j$ be a subset of $V'$ containing $N/(2E)$ of these other strings.
\end{proof}

\begin{proof}[Proof of  Lemma~\ref{lem:stratH}.] 
  Let $\delta = (1-2\varepsilon)/5$. 
  Let $M_{uv} = \max \{X_u, X_v\}$.
  Alice will create pairs $(u,v)$ in which $Z_{uv}/M_{uv}$ oscillates between~$2\delta$ and~$1-2\delta$. 
  The distance between these values is $\delta + 2\varepsilon$, thus to satisfy 
  requirement~\eqref{eq:requirement_eps}, the sum in~\eqref{eq:requirement_a} 
  in such an oscillation increases by at least~$\delta$. 
  (In fact, each cycle contributes at least $2\delta$, but we do not optimize the constant factors.)

  Initially, let $n' = n-1$, and let $U'$ and $ V'$ be disjoint parts of $\{0,1\}^n$ of size~$2^{n-1}$.
  The strategy is recursive in~$n'$.  
  At the start of each recursive call,  we have $X_u \ge n'$ for all~$u \in U'$, and Alice will not have played on the grid $\mathrm Z$ in a row with index smaller than~$n'$. 
  In the beginning of the game, these conditions are trivially satisfied.

  \bigskip
  \begin{samepage}
  {\em Alice's recursive strategy inside disjoint sets $U'$ and $V'$ of size $2^{n'}$.} 

  If $\floor(\delta n') < \sqrt{n}$, the strategy terminates.
  If $n'< n-1$, then for all $u \in U' \cup V'$, Alice places a token in $\mathrm X_u$ at height~$n'+1$. 
  (This guarantees $M_{uv} \le n'+1$, but is not needed in the first recursive call, when $n' = n-1$, since $X_u \le n$ by definition.)
  Then it is Bob's turn.
  If he does not satisfy requirements \eqref{eq:requirement_eps} and  \eqref{eq:requirement_a}, 
  Alice wins and the game terminates. Assume the game continuous.
  
  Let $d = \floor(\delta n')$.
  Let $(U_1, V_1), \ldots, (U_E, V_E)$ be a sequence that satisfies the conditions of  Lemma~\ref{lem:main_help} with $N = 2^{n'}$ and $E = 2^{n'-d}$. 
  (These sets have size $2^{d-1}$.)
  For all $j \le E$ and all pairs $(u,v) \in U_j \times V_j$, Alice places a token in $\mathrm Z_{uv}$ at height $d$. 
  Then it is Bob's turn.
  If he does not satisfy requirements \eqref{eq:requirement_eps} and  \eqref{eq:requirement_a}, 
  the game and the strategy terminates.

  Otherwise, Alice selects some index $j \le E$ for which $X_u \ge \log E$ for all $u \in U_j$. 
  (Such $j$ exists, because there are less than $E$ strings $u$ with $X_u < \log E$, and the sets $U_1, \ldots, U_E$ are pairwise disjoint.)
  She runs the strategy recursively for $U' \leftarrow U_j$, $V' \leftarrow V_j$ and $n' \leftarrow d-1$.
  {\em End of the strategy.} 
  \end{samepage}

  \bigskip
  \noindent
  We need to prove that for some $\rho > 0$, Alice wins the game with parameter $a = \rho \log n$ for large~$n$.
  Assume $d = \floor(\delta n') \ge \sqrt{n}$.
  We show that the total update in the selected set $U_j \times V_j$ increases by at least~$\delta$.
  After Alice's first move, for all $(u,v) \in U' \times V'$, we have 
  \[
    \frac{Z_{uv}}{M_{uv}} \;\ge\; \frac{(1-\delta)n'}{n'+1} \;\ge\; 1-2\delta \,.
  \]
  If Bob's reply does not satisfy the requirements, we are finished.
  Otherwise, Alice performs her second move. 
  For all $j \le E$ and $(u,v) \in U_j \times V_j$, we have
  \[
    \frac{Z_{uv}}{M_{uv}} \;\le\; \frac{\delta n'}{\ceil((1-\delta)n')} \;\le\; 2\delta,
  \]
  where the right inequality follows from our choice of~$\delta \le 1/4$ and for large~$n'$.
  As explained above, if Bob satisfies the requirements, then the total update of $f$ increases by at least~$\delta$.
  
  We now determine a value of $\rho$ such that Alice wins the game~$\mathrm H_{n, \varepsilon , \rho \log n}$ for large~$n$.
  Except for the last, each recursive call increases the total update 
  by at least $\delta$, and the number of such calls
  is 
  \[
    r \;=\; \log_{2/\delta} \frac{\delta \cdot (n-1)}{2\sqrt{n}}. 
  \] 
    Note that the base of the logarithm is $2/\delta$, 
  because for $n \ge 16$, the assumption $d-1 = \floor(\delta n') - 1 \ge \sqrt{n}$ implies $\floor (\delta n') - 1 \ge \delta n'/2$.
  Hence Alice wins the game for $\rho$ and $n$ such that $\rho \log n \le \delta \cdot r$. 
  The lemma follows for any $\rho$ arbitrarily close to
  \[
    \frac{\delta}{2 \log \tfrac{2}{\delta}} \;=\;  \frac{(1-2\varepsilon)}{10\log \tfrac{10}{1-2\varepsilon}}.
    \qedhere
  \]
\end{proof}

\section{Oscillations of $0$-approximations, winning strategy}

\begin{lemma}\label{lem:stratG}
  There exists a constant $\gamma > 0$ such that for all~$c \ge 3$ and~$n> 16c^2$, Alice has a winning strategy in the game~$\mathcal G_{n, c, \gamma n/c}$\,. 
\end{lemma}

\noindent
Together with  Lemma~\ref{lem:zeroApproximationGivesWinning} this implies  Theorem~\ref{th:oscillations}.

In the winning strategy from the previous section,
we obtain a logarithmic number of oscillations. 
To increase this number, we must decrease $n'$ by a smaller amount: by a constant for the $\Omega(n)$ lower bound of $\nid'$ 
and logarithmic for the $\Omega(n/\log n)$ of $\nid$.
Again we will consider a set $S = \bigcup_{j \le E} U_j \times V_j$, but now $E$ will be very small, and it will 
no longer be possible to achieve the requirements for $X_u$ and $Z_{uv}$ for all pairs in $U_j \times V_j$ for some~$j$.
However, we can achieve that the average number of oscillations grows proportional with the number of recursive calls, and this is enough for our purposes.


For a finite set $S$, and a function $h$ on $S$, let $\avg_S [h]$ denote the expected value of $h(s)$ when $s$ is uniformly distributed over~$S$.
%
We present 2 technical and trivial lemmas that are useful for later reference.

\noindent
\begin{lemma}\label{lem:highRatioG}
  Let 
  \\- $N$ and $E\le N$ be non-negative powers of~$2$ with $E \le N/2$,
  \\- $U_1, \ldots, U_E$ be any partition of $U'$ into $E$ sets of size~$N/E$, 
  \\- $a \colon U' \times V' \rightarrow \mathbb R$.
  \\There exist subsets $V_1, \ldots, V_E$ of $V'$ of size $N/E$ such that for $S = \bigcup_{j \le E} (U_j \times V_j)$ we have
  \[ 
    \avg_{S} \big[ a \big] \;\ge\; \avg_{U' \times V'} \big[a\big].
  \] 
\end{lemma}

\begin{proof}
  This follows by the probabilistic method using a uniformly random selection of the sets~$V_1, \ldots, V_E$.
\end{proof}

\noindent
For an integer $k$, let $[X \ge k]$ denote the function that maps a string $u$ to $1$ if $X_u \ge k$ and to $0$ otherwise. Similar for $[Z \ge k]$.

\begin{lemma}\label{lem:oscillations_help} Let $g \colon U' \rightarrow [0,1]$ and $h \colon U' \times V' \rightarrow [0,1]$. If the players satisfy the row restion, then
\begin{align}
  \avg_{U'} \; [g \cdot [X \ge n'-i]] \;&\ge\; \avg_{U'} \; [g] - 2^{-i} \tag{\texttt{EX}}\label{eq:avgOver_X} \\
  \avg_{U' \times V'} [h \cdot [Z \ge n'-i]] \;&\ge\; \avg_{U' \times V'} [g] - 2^{-i}. \tag{\texttt{EZ}}\label{eq:avgOver_Z} 
\end{align}
\end{lemma}

\begin{proof}[Proof of  Lemma~\ref{lem:stratG}.] 
  The idea is to create pairs $(u,v)$ in which $Z_{uv}/M_{uv}$ oscillates between $\tfrac{n-3}{n}$ and $\tfrac{n-3c-9}{n-3}$. 
  Again we initialize $n' = n-1$, and let $U'$ and $V'$ be disjoint sets of $\{0,1\}^n$ of size $2^{n'}$. The strategy is recursive in~$n'$. 

  \bigskip
  \begin{samepage}
  {\em Alice's recursive strategy inside disjoint sets $U'$ and $V'$ of size $2^{n'}$, started at round~$t$.} 

  Let $e = 4c$. If $n' \le e$, then the strategy immediately terminates.
  Assume $n' \ge e+1$.
  If $n' < n-1$, Alice places a token in $\mathrm X_u$ at height $n'+1$ for all $u \in U' \cup V'$.
  Then she waits for Bob's reply. If he does not satisfy his requirements, the strategy terminates. Otherwise, 
  the game proceeds to round~$t+1$.

    Let $(U_1, V_1), \ldots, (U_E, V_E)$ be a sequence that satisfies the conditions of Lemma~\ref{lem:highRatioG} for $E = 2^{e}$, $N = 2^{n'}$ 
  and a function $a$ that we determine later. 
  For all $j \le E$ and all pairs $(u,v) \in U_j \times V_j$, Alice places a token in $\mathrm Z_{uv}$ at height $n' - e$. 
  Then it is Bob's turn.
  If he does not satisfy the requirements, the strategy terminates. Otherwise, the game proceeds to round $t+2$. 

  Alice selects an index $j \le E$ such that the average in \eqref{eq:avgS} below is not smaller than the average of this function over~$U_j \times V_j$.
  Then she runs the strategy recursively 
  for $U' \leftarrow U_j$, $V' \leftarrow V_j$ and $n' \leftarrow n'-e$.
  \\{\em End of the strategy.} 
  \end{samepage}

  \bigskip
  \noindent
  Let $X^{(t)}_u$ and $Z^{(t)}_{uv}$ represent the values of the column at the end of round~$t$.
  We show that except for the first and the last, in each recursive calls starting in a round~$t$ with parameter~$n'> e$, 
  the following claim holds:
  \begin{claim*}
      If $X^{(t-1)}_u \ge n'+e-2$, 
      $Z^{(t)}_{uv} \ge n'-2$ and $X^{(t+1)}_u \ge n'-2$ then $\osc_{t+1}(uv) \ge 1 + \osc_{t-1}(uv)$.
  \end{claim*}

  \begin{proof}[Proof of the claim.]
  Note that Bob must satisfy requirement~\eqref{eq:requirement_c} in rounds $t-1, t$ and $t+1$, since otherwise there is no next 
  recursive call, contrary to what we assumed.
  At the end of round~$t-1$, 
    we decreased $Z_{uv}$, thus $Z^{(t)}_{uv} \le (n'+e) - e = n'$. Together with $M_{uv} \ge X_u$ and the assumption of the claim, this implies
  \[
    f(uv, t-1) \;<\; \frac{Z^{(t-1)}_{uv}+c}{M^{(t-1)}_{uv}} \;\le\; \frac{n'+c}{n'+e-2}. 
  \]
  After Alice's first move in the recursive call, we have $M_{uv} \le n'+1$. Thus by the assumption of the claim,
  \[
    \frac{n'-3}{n'+1+c} \;\le\; \frac{Z^{(t)}_{uv}-1}{M^{(t)}_{uv}+c} \;<\; f(uv, t).
  \]
  After Alice's second move, we have~$Z_{uv} \le n'-e$. Thus by the assumption of the claim,
  \[
    f(uv, t+1) \;<\; \frac{Z^{(t+1)}_{uv}+c}{M^{(t+1)}_{uv}} \; \le \; \frac{n'-e+c}{n'-2}. 
  \]
  One may calculate that for $e = 4c$ and $c \ge 3$, this implies $f(u,v,t-1) < f(u,v,t)$ and $f(u,v,t+1) < f(u,v,t)$.
    (This only needs to be checked for the worst-case values $c = 3$ and $n' = 13 > e$.)
  The claim is proven.
  \end{proof}

  \medskip
  \noindent
  To show that the strategy wins the game, we prove that under the assumptions of the claim, we have:
  \begin{equation}\tag{*}\label{eq:goal}
    \avg_{U_j \times V_j} \Big[ \osc_{t+1} + [X^{(t+1)} \ge n' - 2]\Big]  \;\,\ge\,\; \frac{1}{4} \,+\, \avg_{U' \times V'} \Big[ \osc_{t-1} + [X^{(t-1)} \ge n' + e - 2] \Big] \,,
  \end{equation}
  where $j$ is the index selected in the strategy. 
  Recall that this average is taken uniformly over the set of  {\em ordered} pairs~$U_j \times V_j$.
  Note that an expected value over an indicator function is at most~$1$.
  Since the strategy can execute a linear number of rounds in~$n$, this implies that $f$ makes a linear number of oscillations.

  We prove the inequality in several steps. 
  Let $I_1 = [X^{(t-1)} \ge n' + e - 2]$, which is the first assumption inside the claim. 
  We start with the expectation in the right-hand side and subtract~$1/4$:
  \[
    \avg_{U' \times V'} \Big[ \osc_{t-1} + I_1 \Big] \;-\; \frac{1}{4}. 
  \]
  We apply inequality \eqref{eq:avgOver_Z} above for $h = I_1$:
    \[
      \le \,\;\avg_{U' \times V'} \Big[ \osc_{t-1} + I_1 \cdot [Z^{(t)} \ge n' - 2] \Big].
    \]
  Intuitively, this equation expresses that after the first move of Alice, Bob can not decrease $Z_{uv}$ on too many pairs $(u,v)$.
  We choose $S$ as in Lemma~\ref{lem:highRatioG} with $a$ being the function inside this expectation, 
  i.e., $a = \osc_{t-1} +  I_1\cdot I_2$ where $I_2 = [Z^{(t)} \ge n' - 2]$ is the second condition inside the claim.
  Thus, 
    \[
      \le \,\;\avg_{S} \Big[ \osc_{t-1} + I_1 \cdot I_2\Big].
    \]
  In her second move, Alice decreases $Z_{uv}$ for all $(u,v) \in S$. The following inequality expresses that 
  Bob can not respond by decreasing too many values of $X$. Note that after projecting a uniformly random pair in $S$ to its first coordinate, 
  we obtain the uniform distribution in~$U'$, because the sets $U_1, \ldots, U_E$ form a partition of~$U'$. 
  Similarly, we can write an expectation over $S$ as an expectation of some function over $U'$, and apply~\eqref{eq:avgOver_X}.
  We obtain the following 
  \[
    \le \;\, \avg_{S} \Big[ \osc_{t-1} + I_1 \cdot I_2 \cdot [X^{(t+1)} \ge n' - 2] \Big] \;+\; \frac{1}{4}.
  \]
  Let $I_3 = [X^{(t+1)} \ge n' - 2]$ be the 3rd condition of the claim. Since $\avg_{U'} [I_3] \ge 3/4$, we have
  \[
    \le \;\, \avg_{S} \Big[ \osc_{t-1} + I_1 \cdot I_2 \cdot I_3 + I_3 \Big] \;+\; \frac{1}{4} \;-\; \frac{3}{4}. 
  \]
  By the claim this is
  \begin{equation}\tag{**}\label{eq:avgS}
    \le \;\, \avg_{S} \Big[ \osc_{t+1} + I_3\Big] \;-\; \frac{1}{2}.
  \end{equation}
  Finally, Alice selects $j$ and hence the pair of subsets $(U_j, V_j)$, such that this expectation does not decrease.
  Thus, in the above inequality, we may replace the average over $S$ by the average over $U_j \times V_j$.
  After rearranging the additive constants, we obtain the required inequality~\eqref{eq:goal}.
  The lemma and hence also  Theorem~\ref{th:oscillations} is proven.
\end{proof}

\bibliographystyle{plain}
\bibliography{kolmogorov}

\begin{thebibliography}{1}

\bibitem{NIDoscillations}
Klaus Ambos-Spies, Wolfgang Merkle, and Sebastiaan~A Terwijn.
\newblock Normalized information distance and the oscillation hierarchy.
\newblock {\em arXiv preprint arXiv:1708.03583}, 2017.

\bibitem{idRevisited}
Bruno Bauwens.
\newblock Information distance revisited.
\newblock 2018.

\bibitem{compcomp}
Bruno Bauwens and Alexander Shen.
\newblock Complexity of complexity and maximal plain versus prefix-free
  {K}olmogorov complexity.
\newblock {\em Journal of Symbolic Logic}, 79(2):620--632, 2013.

\bibitem{infoDistance}
Charles~H Bennett, P{\'e}ter G{\'a}cs, Ming Li, Paul~M.B. Vit{\'a}nyi, and
  Wojciech~H Zurek.
\newblock Information distance.
\newblock {\em IEEE Transactions on information theory}, 44(4):1407--1423,
  1998.

\bibitem{complexityOfComplexity}
P.~G{\'a}cs.
\newblock On the symmetry of algorithmic information.
\newblock {\em Soviet Math. Dokl.}, 15:1477--1480, 1974.

\bibitem{LiVitanyiForthEdition}
Ming Li and Paul~M.B. Vit{\'a}nyi.
\newblock {\em An Introduction to {K}olmogorov Complexity and Its Applications,
  4th edition}.
\newblock Springer, 2019.

\bibitem{KolmogorovGames}
Andrei~A. Muchnik, Ilya Mezhirov, Alexander Shen, and Nikolay Vereshchagin.
\newblock Game interpretation of {K}olmogorov complexity.
\newblock unpublished, mar 2010.

\bibitem{bookShenVereshchagin}
Alexander Shen, Vladimir~A Uspensky, and Nikolay Vereshchagin.
\newblock {\em Kolmogorov complexity and algorithmic randomness}, volume 220.
\newblock American Mathematical Soc., 2017.

\bibitem{NIDnonapprox}
Sebastiaan Terwijn, Leen Torenvliet, and Paul~M.B. Vit{\'a}nyi.
\newblock Nonapproximability of the normalized information distance.
\newblock {\em Journal of Computer and System Sciences}, 77:738--742, 2011.

\end{thebibliography}

\appendix

\section{Definition of Kolmogorov complexity}
\label{sec:introKolm}

\subsection{Plain complexity}
\label{ss:introC}

Given a Turing machine $M$ that maps pairs of strings to strings, let
\[
  \C_M(x \mid y) \; = \; \min \left\{ \textnormal{length}(p) : M(p,y) = x \right\},
\]
and let $\C_M(x) = \C_M(x \mid \text{empty string})$.
A Turing machine $U$ is {\em optimal} if for every other Turing machine~$M$ there exists a constant~$c$ 
such that $\C_U(\cdot \mid \cdot) \le \C_M(\cdot \mid \cdot) + c$.
We fix an optimal machine $U$ and write $\C(\cdot \mid \cdot) = \C_U(\cdot \mid \cdot)$.
We define the complexity of an integer $n$ by associating it to the 
string containing $n$ zeros.
More generally, we define the complexity for tuples, sets and other objects, by associating them to strings 
in some computable way.

\medskip
\noindent
 {\em Simple properties.} 
\begin{itemize}[leftmargin=*]
  \item 
    For every $n$-bit string $x$, we have $\C(x) \le n + O(1)$, because we can consider the trivial machine 
    that halts immediately, and for which $M(x,y) = x$ for all $x$ and~$y$. 
    We have $\C_M(x) = n$, and by optimality of $U$, we have $\C(x) \le n+O(1)$.

  \item 
    If $f$ is a computable function, then $\C(f(x)) \le \C(x) + O(1)$.
    Hence, if $y$ has length $n$ and every bit at an odd position is zero, 
    then $\C(y) \le n/2 + O(1)$.

  \item 
    For every integer, we have $\C(n) \le \log n + O(1)$, since each positive integer has a binary representation 
    of size $\log n + 1$.

  \item 
    For all $n$, there exists an $n$-bit that satisfies $\C(x) \ge n$. 
    Indeed, the number of programs less than $n$ is at most $2^0 + 2^1 +\ldots + 2^{n-1} = 2^n - 1$, 
    hence, there must be some string $x$ that has no such program.
    More generally, the fraction of $n$-bit strings $x$ with $\C(x) \ge n-k$ exceeds $1-2^{-k}$.
\end{itemize}

\begin{lemma}\label{lem:Kolm_not_computable}
 The function $\C(\cdot)$ is not computable. 
\end{lemma}

This result follows from an argument similar to Berry's paradox, which 
considers: ``The smallest number that can not be described in less than 20 words''. If this number is well defined, then 
than the statement is a contradiction. 

\begin{proof}
  Consider
 \[
   B(n) \;=\; \min \left\{B' : \C(B') \ge n \right\}.
 \]
 If $\C$ where computable, then also $B$ would be computable. Hence for all $n$:
 \[
   \C(B(n)) \;\le\; \C(n) + O(1) \;\le\; O(\log n). 
 \]
 This contradicts the definition of $B(n)$ for large~$n$. 
\end{proof}

\subsection{Prefix complexity}
\label{ss:introK}

A Turing machine $M$ is {\em prefix-free} if for every pair $(p,y)$ such that $M(p,y)$ halts, 
there exist no strings $q$ that have prefix $p$ and for which $M(q,y)$ halts. In other words, 
for each $y$, the set $\{p : M(p,y) \text{halts}\}$ is a prefix-free set.

There exist optimal prefix-free machines $U$: for every other prefix-free machine $M$, 
there exists a constant $c$ such that $\C_U(\cdot \mid \cdot) \le \C_M(\cdot \mid \cdot) + c$.
We fix such an optimal machine $U$, and write $\K(x \mid y) = \C_U(x\mid y)$. 

\medskip
\noindent
 {\em Simple properties.} 
\begin{itemize}[leftmargin=*]
  \item $\K(x \mid y) \ge \C(x \mid y) - O(1)$.
  \item 
    For every $n$-bit string $x$, we have $\K(x\mid n) \le n + O(1)$. Indeed, consider the mapping
    $M(p,n) = p$ if $p$ has length $n$ and is undefined otherwise. This mapping defines 
    a prefix-free machine, and the result follows by optimality of~$U$.

  \item 
    $\K(x,y) \le \K(x) + \K(y \mid x) + O(1)$. This holds, by considering the concatenation 
    of programs for the universal machine.

  \item 
    $\K(y) \le n + O(\log n)$ for every $n$-bit $y$.
    This holds by the previous item for $x = n$, and $\K(n) \le O(\log n)$.

  \item 
    The function $\K(\cdot)$ is non-computable as well, for the same reasons as for plain complexity.
\end{itemize}

\subsection{The information distance}
\label{ss:infoDist}

This distance was defined in~\cite{infoDistance} as $\id(x,y) = \max \{\K_U(x \mid y), \K_U(y \mid x)\}$, 
where $U$ is a machine that makes this distance minimal up to an additive constant.  
Let $c$ be a constant and consider the function 
\[
 D(x,y) = 
   \begin{cases}
     \id(x,y)+c & \text{ if } x \not= y \\
     0  & \text{ if } x= y.
   \end{cases}
\]
$D$ satisfies the axioms of a metric for a large~$c$:
\begin{itemize}
  \item $D(x,y) \ge 0$,
  \item $D(x,y) = 0$ if and only if $x=y$,
  \item $D(x,y) = D(y,x)$,
  \item $D(x,z) \le D(x,y) + D(y,z)$, (triangle inequality).
\end{itemize}

\noindent
The last property follows from
\[
  \K(x \mid z) \;\le\; \K(x \mid y) + \K(y \mid z) + O(1),
\]
and the symmetric inequality for $\K(z \mid x)$, by setting $c$ equal to 
the $O(1)$ constant (which only depends on the choice of $U$); this inequality follows by concatenating programs.

In~\cite{infoDistance}, it is shown that for a suitable\footnote{
  It is not enough that $U$ is an optimal prefix-free Turing machine, as explained in~\cite[Propositions 1 and 2]{idRevisited}.
  }
optimal prefix free machine $U$, we have that 
\[
  \id(x,y) \;=\; \min \{\text{length}(p) : U(p,x) = y \text{ and } U(p,y) = x \} + O(\log \id(x,y)).
\]
In~\cite{idRevisited} it is shown that the logarithmic precision can be improved to $O(1)$ for strings $x,y$ 
of length $n$ with $\id(x,y) \ge 6\log n$, but can not be improved for strings that have at most logarithmic distance.

\end{document}